\documentclass[10pt,conference,letterpaper]{IEEEtran}

\usepackage{cite}
\usepackage{amsmath,amssymb,amsfonts}
\usepackage{booktabs}
\usepackage{array}
\usepackage{tabularx}
\usepackage{graphicx}
\usepackage{xcolor}
\usepackage{url}
\usepackage{float}
\usepackage{algorithm}
\usepackage{algorithmic}
\usepackage{microtype}
\usepackage{tikz}
\usetikzlibrary{arrows.meta,positioning,fit,calc}

\usepackage{balance}

\definecolor{asgeblue}{HTML}{0072B2}
\definecolor{asgered}{HTML}{B2182B}
\definecolor{asgegreen}{HTML}{009E73}
\definecolor{asgegray}{HTML}{5B616B}
\definecolor{asgegrid}{HTML}{D8DDE5}
\definecolor{asgeink}{HTML}{24272C}
\colorlet{asgeaccent}{asgered}

\tikzset{
  pics/asge document/.style={code={
    \path[draw=asgegray!82,fill=white,line width=.55pt,rounded corners=.35mm]
      (-2.35,-2.75)--(2.35,-2.75)--(2.35,1.35)--(.95,2.75)
      --(-2.35,2.75)--cycle;
    \draw[asgegray!70,line width=.45pt] (.95,2.75)--(.95,1.35)--(2.35,1.35);
    \draw[asgegray!60,line width=.45pt]
      (-1.35,.75)--(.35,.75)
      (-1.35,-.15)--(1.35,-.15)
      (-1.35,-1.05)--(.75,-1.05);
  }},
  pics/asge verifier/.style={code={
    \path[draw=asgeblue,fill=asgeblue!9,line width=.65pt,rounded corners=.25mm]
      (0,3.0)--(2.45,2.05)--(2.05,-.55)
      .. controls (1.75,-1.85) and (.75,-2.55) .. (0,-2.95)
      .. controls (-.75,-2.55) and (-1.75,-1.85) .. (-2.05,-.55)
      --(-2.45,2.05)--cycle;
    \draw[asgeblue,line width=.75pt,line cap=round,line join=round]
      (-1.05,-.05)--(-.25,-.95)--(1.25,.85);
  }},
  pics/asge memory/.style={code={
    \path[draw=asgegray!82,fill=asgegray!7,line width=.6pt]
      (-2.45,-1.95)--(-2.45,1.8)
      arc[start angle=180,end angle=360,x radius=2.45,y radius=.78]
      --(2.45,-1.95)
      arc[start angle=0,end angle=-180,x radius=2.45,y radius=.78]--cycle;
    \path[draw=asgegray!82,fill=asgegray!9,line width=.6pt]
      (0,1.8) ellipse[x radius=2.45,y radius=.78];
    \draw[asgegray!72,line width=.45pt]
      (-2.45,-.05) arc[start angle=180,end angle=360,x radius=2.45,y radius=.72];
  }},
  pics/asge model/.style={code={
    \path[draw=asgegray!82,fill=white,line width=.6pt,rounded corners=.5mm]
      (-2.45,-2.45) rectangle (2.45,2.45);
    \foreach \p in {-1.45,0,1.45} {
      \draw[asgegray!74,line width=.45pt] (\p,2.45)--(\p,3.15);
      \draw[asgegray!74,line width=.45pt] (\p,-2.45)--(\p,-3.15);
      \draw[asgegray!74,line width=.45pt] (2.45,\p)--(3.15,\p);
      \draw[asgegray!74,line width=.45pt] (-2.45,\p)--(-3.15,\p);
    }
    \draw[asgegray!78,line width=.5pt]
      (-1.25,-.75)--(-.45,.65)--(.35,-.15)--(1.25,1.0);
    \foreach \p in {(-1.25,-.75),(-.45,.65),(.35,-.15),(1.25,1.0)}
      \fill[asgegray!78] \p circle (.3mm);
  }},
  pics/asge gateway/.style={code={
    \path[draw=asgegray!82,fill=white,line width=.6pt,rounded corners=.7mm]
      (-3.05,-2.35) rectangle (3.05,2.35);
    \draw[asgegray!78,line width=.55pt,-{Latex[length=1.15mm]}]
      (-2.05,.65)--(1.8,.65);
    \draw[asgegray!78,line width=.55pt,-{Latex[length=1.15mm]}]
      (2.05,-.65)--(-1.8,-.65);
  }},
  pics/asge server/.style={code={
    \path[draw=asgegray!82,fill=white,line width=.6pt,rounded corners=.45mm]
      (-2.55,-3.0) rectangle (2.55,3.0);
    \foreach \yy in {-1.55,0,1.55} {
      \draw[asgegray!65,line width=.45pt] (-1.65,\yy)--(1.65,\yy);
      \fill[asgegray!72] (-1.75,\yy+.48) circle (.28mm);
    }
  }},
  pics/asge server selected/.style={code={
    \path[draw=asgeblue,fill=asgeblue!8,line width=.7pt,rounded corners=.45mm]
      (-2.55,-3.0) rectangle (2.55,3.0);
    \foreach \yy in {-1.55,0,1.55} {
      \draw[asgeblue!72,line width=.45pt] (-1.65,\yy)--(1.65,\yy);
      \fill[asgeblue] (-1.75,\yy+.48) circle (.28mm);
    }
  }},
  pics/asge relay/.style={code={
    \path[draw=asgegray!82,fill=white,line width=.6pt,rounded corners=.55mm]
      (-2.65,-2.65) rectangle (2.65,2.65);
    \draw[asgegray!78,line width=.55pt,-{Latex[length=1.05mm]}]
      (-1.75,.7)--(1.55,.7);
    \draw[asgegray!78,line width=.55pt,-{Latex[length=1.05mm]}]
      (1.75,-.7)--(-1.55,-.7);
  }},
  pics/asge relay selected/.style={code={
    \path[draw=asgeblue,fill=asgeblue!8,line width=.7pt,rounded corners=.55mm]
      (-2.65,-2.65) rectangle (2.65,2.65);
    \draw[asgeblue,line width=.6pt,-{Latex[length=1.05mm]}]
      (-1.75,.7)--(1.55,.7);
    \draw[asgeblue,line width=.6pt,-{Latex[length=1.05mm]}]
      (1.75,-.7)--(-1.55,-.7);
  }},
  pics/asge controller/.style={code={
    \path[draw=asgeblue,fill=asgeblue!8,line width=.7pt,rounded corners=.6mm]
      (-3.5,-2.45) rectangle (3.5,2.45);
    \draw[asgeblue!75,line width=.5pt]
      (-2.25,1.1)--(2.25,1.1)
      (-2.25,0)--(2.25,0)
      (-2.25,-1.1)--(2.25,-1.1);
    \fill[asgeblue] (-.9,1.1) circle (.4mm);
    \fill[asgeblue] (1.1,0) circle (.4mm);
    \fill[asgeblue] (-.1,-1.1) circle (.4mm);
    \draw[asgeblue,line width=.55pt] (-.9,.55)--(-.9,1.65);
    \draw[asgeblue,line width=.55pt] (1.1,-.55)--(1.1,.55);
    \draw[asgeblue,line width=.55pt] (-.1,-1.65)--(-.1,-.55);
  }},
  pics/asge lock/.style={code={
    \path[draw=asgeblue,fill=asgeblue!9,line width=.55pt,rounded corners=.35mm]
      (-1.55,-1.75) rectangle (1.55,1.05);
    \draw[asgeblue,line width=.65pt,rounded corners=.7mm]
      (-.95,1.05)--(-.95,1.85)
      .. controls (-.95,3.0) and (.95,3.0) .. (.95,1.85)--(.95,1.05);
    \fill[asgeblue] (0,-.25) circle (.35mm);
  }}
}

\newcommand{\asgerr}{\textsc{ASGE-RR}}

\usepackage{amsthm}
\newcommand{\nop}[1]{}
\newtheorem{theorem}{Theorem}

\begin{document}


\title{ASGE-RR: Agentic Service Graph Embedding with Revisable Reservations for Dynamic AI-Agent Calls}

\author{Trond Vatten, Yuming Jiang \\Norwegian University of Science and Technology (NTNU), Trondheim, Norway} 
\maketitle

\begin{abstract}

AI-agent workflows often involve remote calls to models, memory stores, and tools distributed across a network. As execution progresses, these dependency calls collectively form an agentic service graph (ASG). Unlike traditional service requests, many dependency calls are revealed only at runtime. Consequently, allocating resources to a currently visible call may consume capacity later needed by a call from a higher-value workflow. We formulate this challenge as Agentic Service Graph Embedding (ASGE), an online network-control problem that maps runtime-revealed workflow calls to service replicas and network paths under capacity, cost and deadline constraints. We present ASGE-RR, an online ASGE controller with revisable reservations. ASGE-RR protects capacity for likely future calls while enforcing the constraints. ASGE-RR evaluates candidate replica-and-path mappings against predicted workflow continuations and updates reservations as new execution information becomes available. 

We evaluate ASGE-RR using OpenHands and GPT Researcher workflows executed with gpt-5.6-luna and replayed over in two complementary experimental environments, a controlled Docker testbed and a WAN testbed. The investigation shows that all the evaluated AI-agent tasks expose at least one runtime-revealed dependency call that can be steered before connection establishment. 
Exploiting this control point, even though the experimental environments are small-scale, 
ASGE-RR already demonstrates noticeable potential: It completes (up to) 10\% more workflow value than a same-information rolling-horizon controller and a current-call steering controller on the WAN testbed. 
The results suggest that runtime-revealed workflow structure creates a new network control opportunity: protecting resources for likely future calls allows more AI-agent workflows to finish in time. 
\end{abstract}


\section{Introduction}

An AI-agent workflow for a complex task often includes remote calls to services like AI models, memory stores, and tools distributed across a network. Modern agent-serving platforms {\em dynamically invoke} these services {\em as execution unfolds}, allowing workflows to adapt to intermediate results and environmental conditions. We call each such remote invocation a {\em dependency call}. Because those services are commonly replicated across data centers, edge sites, or clusters, the infrastructure operator often has multiple equivalent endpoints for the same dependency call. Its network controller can therefore choose which replica a workflow contacts and which managed paths carry the request and response traffic. Existing work has already viewed such decisions as opportunities to reduce latency, balance load, or improve resource utilization ~\cite{santhanam2024alto,luo2026agentix,wang2026maestro}. Fig.~\ref{fig:architecture} shows this control surface: an AI-agent workflow reveals a dependency call, a gateway asks a network controller to select an eligible service replica and to manage the request and response paths before opening the connection.\begin{figure}[!t]
\centering
\begin{tikzpicture}[font=\scriptsize,>=Latex,x=1mm,y=1mm]
\tikzset{
  boundary/.style={draw=asgegray!55,rounded corners=2pt,line width=.65pt,
    fill=asgegray!2},
  badge/.style={circle,fill=asgeblue,text=white,minimum size=4mm,
    inner sep=0pt,font=\tiny\bfseries},
  tinylabel/.style={font=\tiny,text=asgegray,align=center},
  pics/asge controller red/.style={code={
    \path[draw=asgered,fill=asgered!8,line width=.7pt,rounded corners=.6mm]
      (-3.5,-2.45) rectangle (3.5,2.45);
    \draw[asgered!78,line width=.5pt]
      (-2.25,1.1)--(2.25,1.1)
      (-2.25,0)--(2.25,0)
      (-2.25,-1.1)--(2.25,-1.1);
    \fill[asgered] (-.9,1.1) circle (.4mm);
    \fill[asgered] (1.1,0) circle (.4mm);
    \fill[asgered] (-.1,-1.1) circle (.4mm);
    \draw[asgered,line width=.55pt] (-.9,.55)--(-.9,1.65);
    \draw[asgered,line width=.55pt] (1.1,-.55)--(1.1,.55);
    \draw[asgered,line width=.55pt] (-.1,-1.65)--(-.1,-.55);
  }}
}

\draw[boundary] (0,0) rectangle (29.5,22.5);
\node[font=\scriptsize\bfseries] at (14.75,19.4) {AI-agent workflow};
\pic[scale=.72] at (5.8,10.7) {asge model};
\pic[scale=.76] at (14.7,10.7) {asge document};
\pic[scale=.72] at (24.0,10.7) {asge verifier};
\draw[->,asgegray!74,line width=.75pt] (8.5,10.7)--(11.8,10.7);
\draw[->,asgeblue,line width=1pt] (17.5,10.7)--(21.1,10.7);
\node[badge] at (19.3,14.9) {1};
\node[font=\tiny,align=center] at (5.8,5.5) {model};
\node[font=\tiny,align=center] at (14.7,5.5) {result};
\node[font=\tiny,align=center] at (24.0,5.5) {verifier};

\pic[scale=.9] at (36.7,10.7) {asge gateway};
\node[font=\tiny\bfseries] at (36.7,5.0) {gateway};
\draw[<->,asgeblue,line width=1.05pt] (27.0,10.7)--(33.5,10.7);

\draw[boundary] (43.0,0) rectangle (64.0,18.5);
\node[font=\tiny\bfseries] at (53.5,15.8) {managed network};
\pic[scale=.58] at (46.7,9.3) {asge relay selected};
\pic[scale=.58] at (53.5,12.7) {asge relay selected};
\pic[scale=.58] at (53.5,5.9) {asge relay};
\pic[scale=.58] at (60.3,9.3) {asge relay selected};
\draw[asgegray!62,line width=.72pt]
  (49.0,7.9)--(51.2,6.5)
  (55.8,6.5)--(58.0,7.9);
\draw[asgeblue,line width=1.25pt]
  (49.0,10.7)--(51.2,12.1)
  (55.8,12.1)--(58.0,10.7);
\draw[->,asgeblue,line width=1.05pt] (39.7,10.7)--(43.9,9.3);

\draw[boundary] (67.0,0) rectangle (84.0,22.5);
\node[font=\tiny\bfseries] at (75.5,19.4) {eligible replicas};
\pic[scale=.78] at (75.5,13.2) {asge server selected};
\pic[scale=.78] at (75.5,5.6) {asge server};
\node[font=\tiny\bfseries] at (80.0,13.2) {A};
\node[font=\tiny] at (80.0,5.6) {B};
\draw[->,asgeblue,line width=1.2pt] (63.0,10.5)--(72.8,12.9);
\draw[->,asgegray!62,line width=.72pt] (63.0,8.3)--(72.8,5.9);
\node[badge] at (65.7,11.8) {3};

\node[font=\scriptsize\bfseries] at (53.5,31.0) {network controller};
\pic[scale=1.02] at (53.5,26.4) {asge controller red};
\node[badge,fill=asgered] at (43.0,25.8) {2};
\draw[->,asgered,dashed,line width=1pt]
  (36.7,13.7) .. controls (36.7,21.5) and (42.2,25.0) .. (49.9,25.7);
\draw[->,asgered,dashed,line width=1pt]
  (57.1,25.7) .. controls (63.5,25.2) and (65.8,16.4) .. (72.9,15.2);
\node[font=\tiny,text=black,fill=white,inner sep=.5pt] at (39.2,20.6)
  {call + deadline};
\node[font=\tiny,text=black,fill=white,inner sep=.5pt] at (63.8,22.7)
  {replica + paths};

\end{tikzpicture}
\caption{Conceptual overview of where ASGE acts. (1) An intermediate result reveals a dependency call. (2) Before connection, the gateway asks the network controller for a mapping. (3) The controller selects an eligible service replica and managed request and response paths.}
\label{fig:architecture}
\end{figure}
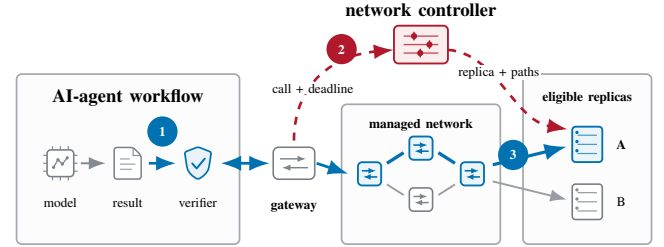

In this paper, we argue that adaptive AI-agent workflows create another Network-for-AI opportunity: network control can increase the amount of workflow value completed before deadline without adding capacity. However, there is an inherent challenge that arises because AI-agent workflows reveal resource allocation demands incrementally at runtime: a workflow's next remote call is often unknown until an earlier stage completes. Consequently, future network resource demands become visible only during execution. 

This runtime revelation creates a new network resource allocation problem, which we call Agentic Service Graph Embedding (ASGE). The problem is to decide how to map a workflow’s agentic service graph (a directed acyclic graph (DAG)) onto network resources, such as service replicas and communication paths. While a workflow runtime determines which call to invoke next, the network decides where the call is served and how traffic reaches it. ASGE differs from established networking formulations, including Service Function Chaining (SFC) ~\cite{halpern2015sfc} and Virtual Network Embedding (VNE) ~\cite{etsi2014nfv}, which assume requests are known at admission and during embedding. In contrast, ASGE operates when a workflow is already running, earlier bindings are fixed, and future dependency calls may not yet exist from the control perspective. ASGE is uniquely challenging because a decision optimal for the currently visible call may consume capacity later required by a higher-value workflow that has not yet been revealed.

To address this challenge, we develop ASGE-RR (ASGE with Revisable Reservations). ASGE-RR evaluates each currently available replica-and-path mapping by asking whether likely future workflow continuations would remain feasible after the decision is made. Rather than permanently reserving and binding resources, ASGE-RR makes {\em revisable reservations} that protect capacity for predicted future calls while allowing those reservations to be updated whenever new workflow information becomes available. In this way, ASGE-RR converts partial knowledge about future workflow structure into admission protection without violating hard capacity constraints. 

We make several contributions. First, we formulate ASGE and its complete-information benchmark. Second, we design ASGE-RR, an online ASGE controller with revisable reservations, and prove that every committed action satisfies resource and policy constraints. In addition, we evaluate ASGE-RR on captured OpenHands and GPT Researcher workflows in Docker and on a WAN. It is found that all the evaluated OpenHands and GPT Researcher tasks expose at least one runtime-revealed dependency call that can be steered before connection establishment. Exploiting this control point, ASGE-RR already demonstrates noticeable potential even under small-scale experiments: ASGE-RR completes (up to) 10\% more workflow value than a same-information rolling-horizon controller and a current-call steering controller on the WAN testbed, without adding network capacity.

\section{Motivation, Model and Benchmark}
\label{sec:model}

\subsection{The problem, challenge and opportunity}
We consider AI-agent workflows that call possibly replicated services across a network managed by an infrastructure operator. For each dependency call, the infrastructure operator can choose an equivalent replica and managed paths to and from it. By coordinating calls that share a workflow deadline, the network may help more valuable workflows finish on time without adding capacity. After completion, the related calls form the workflow's service graph. The replica-and-path selection for each call becomes a mapping of that service graph onto network resources. This is the ASGE problem.
While coordinating resource use among calls opens up an opportunity for the infrastructure operator to make better use of resources, this adaptability also creates a challenge: the next dependency call may appear only after earlier network choices are fixed. 
To demonstrate, consider an AI-agent workflow responding to a network service-level objective (SLO) violation. It collects telemetry, calls model and memory services to diagnose the cause, and asks a planner to propose an action. A verifier may accept the proposal, reject it and trigger more diagnosis, or send the approved action to a tool gateway. The possible steps may be known in advance, but the exact call sequence is not.

An illustration of the challenge and opportunity is provided in Fig.~\ref{fig:late-choice}, where there are two paths from the AI agent to the memory: a ``fast'' path and a ``slow'' path. It also shows two memory-related calls: an ordinary ``memory'' call that is available now, and a memory lookup call from the ``verifier'' that arrives later. For the ordinary memory call, either path can meet its requirement, but the verifier needs to use the fast path to finish by the workflow deadline.

For the challenge, Fig.~\ref{fig:late-choice}a shows that if the replica-and-path mapping is chosen solely based on current information, the fast path would be assigned to the ordinary memory call, causing the later verifier call to fail to meet the workflow deadline. For the opportunity, Fig.~\ref{fig:late-choice}b demonstrates that if there were a way to coordinate such that the ordinary memory call would use the slow path, then the later verifier would be able to use the fast path and both calls would finish on time. 

\begin{figure}[!t]
\centering
\begin{tikzpicture}[font=\scriptsize,>=Latex,x=1mm,y=1mm]
\tikzset{
  panel/.style={draw=asgegray!56,rounded corners=2pt,line width=.65pt,
    fill=asgegray!2},
  lane/.style={asgegray!55,line width=.8pt},
  tinylabel/.style={font=\fontsize{5.4}{6.0}\selectfont,
    text=asgegray!92,align=center}
}

\draw[panel] (0,0) rectangle (40.5,34.5);
\node[font=\scriptsize\bfseries] at (20.25,31.6) {(a) Current-only choice};
\node[tinylabel,anchor=east] at (9.0,22.0) {scarce};
\node[tinylabel,anchor=east] at (9.0,11.2) {alternate};
\draw[lane,rounded corners=1mm]
  (9.8,21.9)--(29.9,21.9)--(32.2,17.5);
\draw[lane,rounded corners=1mm]
  (9.8,11.1)--(29.9,11.1)--(32.2,15.5);
\pic[scale=.78] at (13.4,21.9) {asge memory};
\node[font=\tiny\bfseries,align=center] at (13.4,26.8) {memory now};
\draw[->,asgered,line width=1.8pt,rounded corners=1mm]
  (16.0,21.9)--(29.9,21.9)--(32.2,17.5);
\pic[scale=.72] at (13.4,15.8) {asge verifier};
\node[font=\tiny\bfseries,align=center] at (13.4,12.8) {verifier later};
\draw[->,asgered,dashed,line width=1pt]
  (16.0,16.2) .. controls (21.0,16.5) and (24.2,18.8) .. (27.6,18.8);
\node[font=\large,text=asgered] at (28.5,18.8) {$\times$};
\pic[scale=.78] at (35.1,16.5) {asge server};
\node[font=\fontsize{5.4}{6.0}\selectfont\bfseries,
  text=asgegray!92,align=center] at (20.25,4.2)
  {verifier misses the deadline};

\draw[panel] (43.5,0) rectangle (84,34.5);
\node[font=\scriptsize\bfseries] at (63.75,31.6) {(b) ASGE-RR};
\node[tinylabel] at (63.75,28.8) {revisable reservation};
\node[tinylabel,anchor=east] at (52.5,22.0) {scarce};
\node[tinylabel,anchor=east] at (52.5,11.2) {alternate};
\draw[asgeblue!70,line width=1.05pt,rounded corners=1mm]
  (53.3,21.9)--(73.4,21.9)--(75.9,17.5);
\draw[lane,rounded corners=1mm]
  (53.3,11.1)--(73.4,11.1)--(75.9,15.5);
\pic[scale=.7] at (68.0,25.0) {asge lock};
\draw[asgeblue,dashed,line width=.7pt] (68.0,22.9)--(68.0,21.9);
\pic[scale=.72] at (56.9,21.9) {asge verifier};
\node[font=\tiny\bfseries,align=center] at (56.9,26.3) {verifier later};
\draw[->,asgeblue,line width=1.6pt,rounded corners=1mm]
  (59.4,21.9)--(73.4,21.9)--(75.9,17.5);
\pic[scale=.78] at (56.9,11.1) {asge memory};
\node[font=\tiny\bfseries,align=center] at (56.9,15.9) {memory now};
\draw[->,asgegray!82,densely dashed,line width=1.6pt,rounded corners=1mm]
  (59.5,11.1)--(73.4,11.1)--(75.9,15.5);
\pic[scale=.78] at (78.6,16.5) {asge server selected};
\node[font=\scriptsize\bfseries] at (72.6,25.3) {$\checkmark$};
\node[font=\fontsize{5.4}{6.0}\selectfont\bfseries,align=center] at (63.75,4.2)
  {both calls finish on time};

\end{tikzpicture}
\caption{Conceptual illustration of the challenge and opportunity.}
\label{fig:late-choice}
\end{figure}
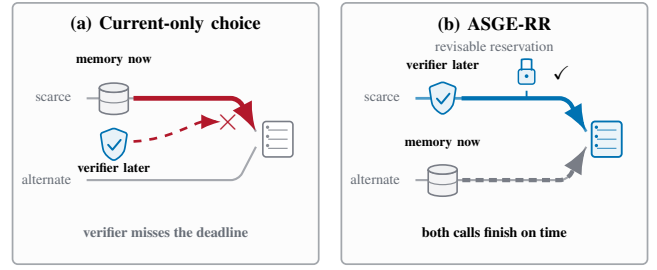

\subsection{System model}
To formally describe and study ASGE, we first introduce the system model in this section and then construct a benchmark in Section~\ref{subsec:benchmark}.

\subsubsection{Network and service substrate}

The network is represented by a directed graph \(G=(V, E)\), where $V$ denotes the set of nodes and $E$ the set of links of the network. Link \(e \in E\) has capacity \(U_e\,[\mathrm{bit/s}]\) and one-way delay \(\delta_e\,[\mathrm{s}]\). For directed path \(p\), let \(\mathbf 1_e(p)\) be one when \(p\) uses link \(e\). Its delay \(\delta(p)\) can be calculated as \(\delta(p)=\sum_{e\in p}\delta_e\). 
Each service type \(d\) may be served by a replica in the set \(\mathcal U_d\). Replica \(u\) resides at node \(n(u) \in V\) and provides service capacity \(Q_u\) measured in service-work units per second [service-work units/s]. 

\subsubsection{Workflow replay}

Note that for an AI-agent task, its exact workflow is available only after execution is complete and depends on network-control decisions. To isolate network control from agent behavior, we use workflow replay: arrivals, branch outcomes, dependency-call releases, request and response sizes, call execution behavior, and realized service times are fixed from the completed execution. This ensures that all policies face identical workflow behavior and allows differences in outcome to be attributed solely to network control. 

AI-agent workflow \(r\) arrives at time \(A_r\) and has deadline budget \(L_r\) with value \(\nu_r\). The absolute deadline is \(A_r+L_r\). 
Let \(\mathcal D_r^\zeta\) denote the set of dependency calls observed in replay \(\zeta\), where repeated invocations remain distinct calls. 

\subsubsection{Delay calculation} 
For call \(c \in \mathcal D_r^\zeta\), let \(v_{rc}\) denote the issuing host, \(d(c)\) the called service type, \(m_{rc}^{\rightarrow}\) [bit] the request size, \(m_{rc}^{\leftarrow}\) the response size [bit], \(b_{rc}^{\rightarrow}\) [bit/s] the reserved transfer rate on the path from the host to the service, and \(b_{rc}^{\leftarrow}\) [bit/s] the reserved transfer rate on the reverse path. A candidate replica-and-path mapping for the call is denoted as
\[
 \gamma=(u,p^{\rightarrow},p^{\leftarrow})
 \in\mathcal G_{rc}(v_{rc}),
\]
which selects replica \(u\in\mathcal U_{d(c)}\), request path \(p^{\rightarrow}\) from
\(v_{rc}\) to \(n(u)\), and response path \(p^{\leftarrow}\), where $\mathcal G_{rc}(v_{rc})$ denotes the set of possible mappings and the two paths \(p^{\rightarrow}\) and \(p^{\leftarrow}\) may differ. 

At the call release, the following information is assumed to be known, e.g. via measurements: whether replica \(u\) is available (\(\bar w_{rcu}=1\)), an upper bound on required service work \(\widehat w_{rcu}\), and an upper bound on setup time \(\kappa_{rcu}\) [s]. Suppose replica \(u\) can support service rate \(\lambda_{rcu}\) [service-work units/s] to the call. Then, these values give the service-time bound 
\(\widehat s_{rcu}=\widehat w_{rcu}/\lambda_{rcu}
+\kappa_{rcu}\). 
Ignoring queueing, the corresponding call completion time is bounded by 
\begin{align}
 d_{rc}(\gamma)={}&d^{ctrl}_{rc}
 +\frac{m_{rc}^{\rightarrow}}{b_{rc}^{\rightarrow}}
 +\delta(p^{\rightarrow})+\widehat s_{rcu}\nonumber\\
 &+\frac{m_{rc}^{\leftarrow}}{b_{rc}^{\leftarrow}}
 +\delta(p^{\leftarrow}),                                           \label{eq:call-delay}
\end{align}
where \(d^{ctrl}_{rc}\) represents the controller and path establishment delay. This delay duration \(d_{rc}(\gamma)\) includes controller processing, data transfer on request and response paths, service execution, and delay of both paths.

\subsection{Complete-information benchmark}
\label{subsec:benchmark}

For a completed replay \(\zeta\), all branch outcomes, dependency
calls, and realized service times are known. Each workflow can
therefore be represented as a fixed, path-based, time-indexed DAG.
We use this information to define a noncausal benchmark: the best
joint replica, path, waiting, and phase schedule on the same finite
time grid used by the online controller.

For workflow \(r\), let \(\mathcal K_r^\zeta\) be the finite set of
complete configurations satisfying replica eligibility, path policy,
stage precedence, and non-preemption. Configuration \(k\) fixes the
replica and paths of every call, any deliberate
waiting, and the request, setup/service, and response phase schedule.
Let \(T_{rk}\) be its arrival-relative completion time and
\(K_{rk}\geq 0\) its total cost under the per-action cost rule
\(C(\cdot)\). Its schedule induces directed-link load
\(B_{rk,e\ell}\) and replica service-rate load \(D_{rk,u\ell}\) in
time bucket \(\ell\). The former includes fixed-route workflow
messages and charges request and response bandwidth only while the
corresponding phase executes. Background loads
\(\bar B_{e\ell}\) and \(\bar D_{u\ell}\) are common to all policies.
Below, \(\sum_{r,k}\) abbreviates
\(\sum_r\sum_{k\in\mathcal K_r^\zeta}\).

Binary \(z_{rk}\) selects configuration \(k\), and \(a_r\) indicates
that workflow \(r\) completes by its deadline. The complete-information
benchmark is
\begingroup
\begin{align}
 \underset{a,z}{\operatorname{lexmax}}\quad
 &\left(
 \sum_r\nu_ra_r,\,
 -\sum_{r,k}T_{rk}z_{rk},\,
 -\sum_{r,k}K_{rk}z_{rk}
 \right)                                                   \label{eq:cfg-objective}\\
 \text{s.t.}\quad
 &\sum_{k\in\mathcal K_r^\zeta}z_{rk}=a_r,
 &&\forall r,                                               \label{eq:cfg-one}\\
 &\bar B_{e\ell}+\sum_{r,k}B_{rk,e\ell}z_{rk}\leq U_e,
 &&\forall e,\ell,                                          \label{eq:cfg-link}\\
 &\bar D_{u\ell}+\sum_{r,k}D_{rk,u\ell}z_{rk}\leq Q_u,
 &&\forall u,\ell,                                          \label{eq:cfg-service}\\
 &\sum_{k\in\mathcal K_r^\zeta}T_{rk}z_{rk}\leq L_ra_r,
 &&\forall r,                                               \label{eq:cfg-deadline}\\
 &z_{rk},a_r\in\{0,1\}.                                     \label{eq:cfg-domain}
\end{align}
\endgroup
The objective first maximizes priority-weighted workflow value
completed by deadline, then minimizes completion time and resource
cost. If each configuration catalog contains every assignment
satisfying the stated per-workflow constraints, 
Eqs.~\eqref{eq:cfg-objective}--\eqref{eq:cfg-domain} give the exact
complete-information optimum for replay \(\zeta\). The benchmark is
not implementable online because future calls and service
times have not yet been revealed.

\subsection{Online ASGE}
The complete-information benchmark assumes knowledge that a live controller lacks, so it is not implementable in practice. In the following, we focus on the online ASGE problem that begins after workflow execution has started and only a prefix of the workflow has been observed.

\subsubsection{Information available at decision time}

The live controller operates from the observed history, remaining capacity, fixed earlier choices, reservations for possible later calls, policy, and readiness. At each event, it first frees capacity that has ended its use. Events include call arrival, fixed stage start/end, branch outcome, call release/return, and termination. If a call appears, it chooses a replica and paths or waits, then books the call's three phases, i.e., the request, service, and response phases. We record all observations in time order. 

For replay \(\zeta\), policy \(\Pi\) produces \(\omega_\Pi(\zeta)=((t_1,e_1),\ldots,(t_T,e_T))\). 
Let \(\mathcal F_\tau=\sigma((t_1,e_1),\ldots,(t_\tau,e_\tau))\) be the history through event \(\tau\). At decision time $\tau$ for a call, the available information or decision state is represented by  
\begin{align}
 I_\tau=(\mathcal F_\tau,\widetilde S_\tau,X_\tau,A_\tau,R_\tau,
          P_\tau,\bar W_\tau),                                      \label{eq:information-state}
\end{align}
where \(\widetilde S_\tau\) records remaining capacity;
\(X_\tau\) contains fixed bindings; \(A_\tau\) contains hard allocations; \(R_\tau\) contains reservations; and \(P_\tau,\bar W_\tau\) represents current policy $P_\tau$ and 
readiness information \(\bar W_\tau\). 

A policy is {\em causal} if decision $d_\tau$ depends only on $I_\tau$:
\begin{align}
 d_\tau=\Pi(I_\tau),                                                  \label{eq:causal-policy}
\end{align}

As a highlight and difference from the benchmark, future observations may influence later decisions here but cannot affect decisions already made. This information asymmetry defines the online ASGE problem. 

\subsubsection{Capacity evolution}
When a decision is committed, the controller binds each call before dispatch and keeps that binding fixed until the response returns. In addition, the controller books resource usage into a capacity calendar that records how available capacity of a service replica or a link evolves. 

Let \(g\) index a time bucket. In bucket \(g\), the maximum available capacity is \(c_{eg}=U_e-\bar B_{eg}\) for link $e$ and \(c_{ug}=Q_u-\bar D_{ug}\) for service replica $u$, where $\bar B_{eg}$ and $\bar D_{ug}$ respectively denote the background load on the link or the replica. To simplify representation, we use \(j \in \{e, u\}\) to index either a link or a service replica in the following.  Let \(\rho_{jg}(e_\tau)\) be load released or expired at the event and \(u_{jg}(d_\tau)\) denote the action's phase calendar, where $j \in \{e, u\}$. 

The transition of load on link $e$ (by letting $j=e$) or  replica $u$ (by letting $j=u$) can be represented by 
\begin{align}
 \widetilde A_{\tau,jg}&=A_{\tau-1,jg}-\rho_{jg}(e_\tau),\qquad
 \widetilde A_{\tau,jg}+u_{jg}(d_\tau)\le c_{jg},                    \label{eq:release}\\
 A_{\tau,jg}&=\widetilde A_{\tau,jg}+u_{jg}(d_\tau),\quad\forall j,g. \label{eq:transition}
\end{align}
The remaining capacity is then 
\(\widetilde S_{\tau,jg}=c_{jg}-\widetilde A_{\tau,jg}\). 

As transactions finish, their allocations are released. Resource updates follow Eqs. (12)-(13), ensuring that committed allocations never exceed available capacity.
As a remark, every allocation uses a half-open phase interval. A call that finishes midway through a bucket releases capacity at the next boundary, which makes rounding conservative. As time advances, expired buckets disappear from the controller's relative view.

\subsubsection{Evaluation metrics} 
As in the complete-information benchmark, we evaluate policies using three quantities: completed workflow value, completion time, and resource cost. For replay \(\zeta\),
the corresponding lexicographic performance vector is defined as 
\begin{align}
 \mathbf J_\Pi(\zeta)=\Big(&
 \sum_r\nu_r\mathbf 1\{T_r^\Pi\le L_r\},\nonumber\\[-1mm]
 &-\sum_{r:T_r^\Pi\le L_r}T_r^\Pi,
 -\sum_\tau C(d_\tau)\Big)                                      \label{eq:trace-objective}
\end{align}
where, if a request is incomplete or dropped, set 
\(T_r^\Pi=\infty\).

Reporting these components separately makes the trade-offs visible: a policy may complete more value, finish workflows sooner, or consume fewer resources.

\subsubsection{Fundamental limitation}

One might hope that a deterministic online controller could guarantee near-optimal performance on every workload. Unfortunately, this is impossible without additional assumptions. No deterministic causal policy for ASGE can admit a positive workload-independent competitive ratio. Recall the example shown in Fig.~\ref{fig:late-choice}, where it is shown that committing capacity to an earlier revealed call blocks a later call from a higher-value workflow revealed after the commitment. A controller with complete information can complete both workflows, whereas any deterministic online policy can be forced arbitrarily far from optimal. 

Consequently, online ASGE cannot rely on worst-case competitiveness. So, instead, we focus on two properties:
\begin{itemize}
    \item hard feasibility guarantees, and
    \item measured performance relative to the complete-information benchmark.
\end{itemize}

The complete-information benchmark defines what is achievable with full knowledge of workflow evolution. The challenge is approaching that benchmark while observing only runtime-revealed calls. The next section presents ASGE-RR, which addresses this challenge through revisable reservations and future-aware admission decisions.

\section{\asgerr: ASGE with Revisable Reservations}
\label{sec:algorithm}

\subsection{Design overview}
ASGE-RR addresses the failure mode in Fig. 2 through revisable reservations: it protects capacity for likely future calls and revises that protection whenever new workflow information arrives. At each call-release event, it performs four steps:
\begin{itemize}
    \item Generate feasible actions: Construct all policy-compliant replica-and-path choices for the newly revealed call.
    \item Forecast plausible workflow continuations: Estimate a bounded set of possible future call sequences that may still occur.
    \item Compute revisable reservations: Determine how much capacity should be protected for those future possibilities.
    \item Select an action: Choose the feasible action that best balances deadline completion, resource cost, and future flexibility.
\end{itemize}

Importantly, only the selected action becomes committed, and only issued calls consume actual resources. However, for resource reservation, ASGE-RR does not permanently reserve capacity; reservations remain in a soft state that guides admission decisions and may change at the next event. As a result, forecast errors may reduce efficiency but cannot violate capacity constraints.

\subsection{Capacity calendar}
ASGE-RR represents link and service-replica capacity using a time-indexed resource calendar. 
Specifically, ASGE-RR divides the time before a workflow deadline into capacity slots. 
Let \(\mathcal J\) index directed links and service-replica resources. 
For each resource $j \in \mathcal J$, which is either a directed link $e$ or a service-replica resource $u$, and time slot $\tau$, the calendar tracks: (1) committed load from issued calls, (2) reserved load for possible future calls, and (3) remaining available capacity.

The calendar is updated every time \(t_\tau\) when a call completes, a new call is revealed, a workflow terminates, or a reservation changes. By maintaining explicit resource occupancy through time, the controller can reason about whether current decisions leave sufficient capacity for future workflow execution. 

When a dependency call is released, ASGE-RR generates all legal mappings. Each candidate mapping specifies 
\[
 \gamma=(u,p^{\rightarrow},p^{\leftarrow}),
\]
and creates a calendar footprint consisting of: request-path bandwidth, replica service usage, and response-path bandwidth. This footprint becomes the basis for later feasibility checks. 

The controller also considers a strategic waiting action, $\gamma_{wait}$,
which postpones dispatch until the next scheduling boundary. Including wait actions is important because immediate service is not always optimal. Under contention, briefly delaying one call can enable multiple workflows to meet their deadlines.

Specifically, the resource allocation uses an absolute time grid with boundaries \(g\Delta\). At time \(t_\tau\), \(g_\tau(1),\ldots,g_\tau(H_{r\tau})\) are the absolute buckets that overlap workflow \(r\)'s remaining interval \([t_\tau,A_r+L_r)\), viewed as relative slots. ASGE-RR charges a partial first bucket in full: \(H_{r\tau}=|\{g:[g\Delta,(g+1)\Delta)\cap[t_\tau,A_r+L_r)\ne\emptyset\}|\); \(H\) below is its workload maximum.

The remaining capacity in this view is \(S_{\tau,j\ell}=c_{j,g_\tau(\ell)}- \widetilde A_{\tau,j,g_\tau(\ell)}\). It already considers issued calls and fixed workflow messages. Candidate \(\gamma\) has a sparse calendar \(u_{\gamma,j\ell}\). It books bandwidth on the request path, \(\lambda_{rcu}\) at the replica during service, and bandwidth on the response path. The action set also includes \(\gamma^{wait}(t^+)\). This action leaves the call undispatched until the first grid boundary at or after the earliest known release \(t^+\), then shifts the candidate calendar to that boundary.

\subsection{Forecasting future calls}
Future calls are unknown at the time of decision. ASGE-RR therefore reasons over a set of plausible workflow continuations. 
We call each plausible workflow continuation or possible remaining call sequence a \emph{suffix}. Examples include: a verifier accepting a plan, a verifier rejecting a plan, a planner generating an additional tool call, a retrieval step causing another model invocation, and workflow termination.

ASGE-RR assigns probabilities to these suffixes using traces collected from previous workflow executions. To keep computation bounded, only the highest-probability suffixes are retained. A configurable coverage threshold determines whether enough probability mass remains. If too much uncertainty exists, ASGE-RR abandons reservations and reverts to a myopic strategy. This fallback mechanism prevents the controller from making aggressive reservation decisions on weak evidence.

Specifically, for workflow $r$, let \(\Sigma_r(q_{r,\tau})\) denote the set of suffixes consistent with the observed workflow state at time $\tau$. ASGE-RR keeps at most \(Z\) suffixes in \(\Sigma_r(q_{r,\tau})\). Their original probability mass is \(m_{r\tau}=\sum_{\sigma\in\Sigma_r}\widehat
p_{r,\tau}(\sigma)\), where \(\widehat
p_{r,\tau}(\sigma)\) denotes the probability of suffix $\sigma$ happening in \(\Sigma_r(q_{r,\tau})\). Let \(\alpha_r\in[0,1]\) be the largest allowed omitted mass. If \(m_{r\tau}<1-\alpha_r\), the forecast omits too much probability and the method falls back to myopic control. Otherwise, it normalizes the retained weights: \(\widetilde p_{r,\tau}(\sigma)=\widehat p_{r,\tau}(\sigma)/m_{r\tau}\). The weights must be learned from traces that are disjoint from the evaluation set.

For each current choice, a bounded search simulates these possible futures.
The search remains causal: futures with the same observed history must make
the same choice until a new observation separates them. If two sequences begin
with the same memory call, for example, they use the same simulated replica
choice until a later verifier result appears. A scenario-tree beam search produces a
future resource calendar \(v_{r\sigma,j\ell}(\gamma)\) and absolute completion estimate \(\widehat C_{r\sigma}(\gamma)\) for each retained suffix, including future
fixed-route workflow messages. A suffix extending beyond slot \(H_{r\tau}\)
receives \(\widehat C_{r\sigma}=\infty\) and holds its capacity bound through
the deadline. We retain candidates that fit every retained scenario.

\subsection{Forecasting future capacity demand and completion}
The retained sequences of \(\Sigma_r(q_{r,\tau})\) are treated as alternatives rather than simultaneous demand. ASGE-RR uses a weighted quantile to choose how much capacity to protect in each resource-time slot and another to estimate completion time. 

Let \(\epsilon_r,\beta_r\in[0,1]\) be the resource and completion tail levels, and
let \(Q_q\) be a weighted \(q\)-quantile over the retained suffixes.
\(R_{r\tau,j\ell}(\gamma)\) summarizes forecast demand for one resource and
slot; \(\widehat C^\beta_{r\tau}(\gamma)\) summarizes forecast completion:
\begin{align}
 R_{r\tau,j\ell}(\gamma)=Q_{1-\epsilon_r}
 \!\left(\{v_{r\sigma,j\ell}(\gamma):
 \sigma\sim\widetilde p_{r,\tau}(\cdot\mid q_{r,\tau})\}\right),      \label{eq:reservation}\\
 \widehat C^\beta_{r\tau}(\gamma)=Q_{1-\beta_r}
 \!\left(\{\widehat C_{r\sigma}(\gamma):
 \sigma\sim\widetilde p_{r,\tau}(\cdot\mid q_{r,\tau})\}\right).
\label{eq:completion-quantile}
\end{align}
\(R_{r\tau,j\ell}(\gamma)\) and \(\widehat C^\beta_{r\tau}(\gamma)\) represent the amount of capacity to be protected for how long for future calls if action $\gamma$ is selected.

Together, these parameters provide a tunable tradeoff between resource efficiency and deadline robustness. Smaller \(\epsilon_r\) and \(\beta_r\) make both tests more conservative. These are marginal per-resource quantiles. At \(\epsilon_r=0\), ASGE-RR
protects the greatest demand for each resource and slot while keeping mutually
exclusive branches separate. Probability mass omitted from \(\Sigma_r\) counts as late in \(\widehat p_{r\tau}^{\mathrm{on}}(\gamma)=
\sum_{\sigma\in\Sigma_r}\widehat p_{r,\tau}(\sigma)
\mathbf 1\{\widehat C_{r\sigma}(\gamma)\le A_r+L_r\}\). Without a usable
template, \(R_{r\tau}=0\), and the completion test uses a configured lower
bound on remaining work. The method is then myopic.

\subsection{Feasibility tests}

The controller next eliminates actions that are unlikely to succeed. A candidate mapping $\gamma$ must satisfy resource feasibility and deadline feasibility. The former ensures that the current load, existing reservations, and new reservations must all fit within available capacity. The latter is to satisfy the estimated completion time. Together, these two tests function as a future-aware admission control mechanism.

Let \(R_{-r,\tau}\) be the capacity protected for other workflows. A current
choice must pass two tests. Its call and all relevant reservations must fit in
every resource-time slot, and its estimated completion must meet the deadline:
\begin{align}
 u_{\gamma,j\ell}+R_{r\tau,j\ell}(\gamma)+R_{-r,\tau,j\ell}
 &\le S_{\tau,j\ell}, &&\forall j,\ell,                              \label{eq:protected}\\
\widehat C^\beta_{r\tau}(\gamma)&\le A_r+L_r.                        \label{eq:deadline-screen}
\end{align}

These forecast tests retain choices predicted to fit and finish. Without a
workflow template, Eq.~\eqref{eq:deadline-screen} uses a lower bound on
remaining work to screen calls whose earliest possible completion exceeds the
deadline.

\subsection{Reservation arbitration}
For workflow \(r\) and feasible candidate mapping \(\gamma\) at time \(\tau\), a reservation priority \(\pi_{r\tau}(\gamma)\) is calculated: 
\begin{align}
 \pi_{r\tau}(\gamma)=\frac{\nu_r\widehat p_{r\tau}^{\,\mathrm{on}}(\gamma)}
 {\max\{\epsilon_t,A_r+L_r-t_\tau\}}                               \label{eq:priority}
\end{align}
where \(\epsilon_t>0\) is a deadline guard measured in seconds.

This priority rises with workflow value, predicted on-time probability, and urgency. The terms prefer earlier completion, lower cost, and less use of scarce capacity. A valuable workflow near its deadline therefore has a higher-priority reservation. If reservations for other workflows block all current actions, ASGE-RR releases lower-priority reservations in increasing order of priority until a feasible action appears, or until no more reservations can be removed. This policy prevents low-priority speculative reservations from blocking high-priority work.

\subsection{Action ranking}
Multiple feasible actions may remain after screening. Rather than choosing the earliest completion time alone, ASGE-RR scores each candidate based on three factors: completion, resource cost, and resource scarcity for future calls. The resulting score balances short-term performance and future flexibility. While a myopic controller focuses almost entirely on the first, ASGE-RR explicitly accounts for all three.

Once an action fits, a score $\operatorname{score}_\tau(\gamma)$ ranks the remaining choices. Let \(K_\tau(\gamma)\) use the same per-action cost rule
\(C(\cdot)\) as the complete-information benchmark to candidate \(\gamma\)'s current-call calendar
\(u_\gamma\) and controller charge. It excludes the soft reservation
\(R_{r\tau}(\gamma)\). Define normalized completion
\(\bar C_{r\tau}(\gamma)=(\widehat C^\beta_{r\tau}(\gamma)-t_\tau)/
\max\{\epsilon_t,A_r+L_r-t_\tau\}\) and normalized cost
\(\bar K_\tau(\gamma)=K_\tau(\gamma)/K_0\). We tune positive scale \(K_0\)
and weights \(\lambda_K,\lambda_R\ge0\) on separate calibration traces.
Per-resource guard \(\eta_j>0\), expressed in resource \(j\)'s units, prevents
division by zero. The score is calculated as
\begin{align}
 \operatorname{score}_\tau(\gamma)={}&\bar C_{r\tau}(\gamma)
 +\lambda_K\bar K_\tau(\gamma)\nonumber\\
 &+\lambda_R\sum_{j,\ell}
 \frac{u_{\gamma,j\ell}+R_{r\tau,j\ell}(\gamma)}
 {\eta_j+S_{\tau,j\ell}-R_{-r,\tau,j\ell}}.                          \label{eq:score}
\end{align}

\subsection{The algorithm and operational behavior}

Algorithm~\ref{alg:asge-rr} summarizes the search, feasibility, arbitration, and commitment procedure. Only the selected action creates hard load; reservations remain soft and are recomputed whenever workflow information changes.

\begin{algorithm}[t]
\caption{\asgerr}
\label{alg:asge-rr}
{\small
\setlength{\algorithmicindent}{.8em}
\begin{algorithmic}[1]
\STATE Precompute \(K\) policy-feasible paths per directed endpoint pair.
\FOR{event \((t_\tau,e_\tau)\) in timestamp order}
  \STATE Release ended use; update traffic, observations, reservations.
  \IF{\(e_\tau\) reveals an undispatched dependency call}
    \STATE Clear \(r\)'s reservation; generate mapping/wait actions.
    \FOR{action \(\gamma\)}
      \STATE Search possible futures; set \(R=0\) if coverage is insufficient.
      \STATE Compute \(R_{r\tau}(\gamma)\), \(\widehat C^\beta_{r\tau}(\gamma)\); reject if a future is infeasible.
    \ENDFOR
    \STATE Keep actions passing Eqs.~\eqref{eq:protected} and \eqref{eq:deadline-screen}.
    \WHILE{no action remains and a reservation is releasable}
      \STATE Release the lowest-priority reservation; test again.
    \ENDWHILE
    \IF{an action remains}
      \STATE Choose minimum-score \(\gamma\); commit its binding or queue wait.
    \ELSE
      \STATE Drop \(r\); no action passes both tests.
    \ENDIF
  \ENDIF
\ENDFOR
\end{algorithmic}
}
\end{algorithm}

\subsection{Safety guarantee and implementation complexity}

Since reservations in ASGE-RR never consume capacity on their own, forecast errors may reduce performance. Still, they cannot violate capacity constraints, which provides a safety guarantee and is summarized in the following theorem. 

\begin{theorem}
    If the initial hard calendar is feasible and realized demands stay within
their declared bounds, every ASGE-RR commitment satisfies capacity and policy
constraints and preserves earlier bindings, independent of forecast accuracy. 
\end{theorem}

\begin{proof}
    Induct on events. Release cannot increase load. Since reservations are
nonnegative, Eq.~\eqref{eq:protected} implies
\(u_{\gamma,j\ell}\le S_{\tau,j\ell}\), so Eq.~\eqref{eq:transition} preserves
capacity. Candidate generation enforces policy; waits pass the same shifted
calendar test; issued bindings never change. Reservations add no hard load,
and an action that still cannot fit is dropped. Declared demand bounds prevent
realized use from exceeding its booking.
\end{proof}

For complexity of Algorithm 1, let \(P_{\mathrm{path}}\) be path-precomputation cost,
\(B_d\) the largest number of eligible replicas, and \(K\) the largest number
of paths in each direction. Then \(B=B_dK^2+1\) bounds the mapping and wait
actions. Let \(Z\) be the number of retained suffixes, \(H\) the number of
time slots, \(W\) the beam width, and \(L\) the number of nonzero
resource-time entries per expansion. One candidate requires at most \(ZHWB\)
beam expansions. Therefore, \(T_d\) call decisions with at most \(R_c\)
reservation-release passes cost
\begin{align}
 O\!\left(P_{\mathrm{path}}+
 T_dR_cZHWB^2[L+\log(WB)]\right).                                    \label{eq:complexity}
\end{align}
Working memory is \(O(ZWHL+NHL+BL)\) for \(N\) active reservation profiles.
Capacity and policy checks are exact for each generated current choice.

\section{Evaluation}
\label{sec:evaluation}

We evaluate ASGE-RR in two complementary experimental environments: a controlled Docker testbed and a five-VM WAN testbed. We run fixed OpenHands and GPT Researcher runtimes and use the resulting AI-agent workflows in both setups.

\subsection{Real AI-agent runtimes and the pre-connection control point}
\label{sec:evaluation.A}

We run fixed OpenHands SDK 1.28.1~\cite{wang2025openhands} and GPT Researcher
3.5.1~\cite{elovic2026gptresearcher} versions on 28 public tasks chosen before seeing model results. OpenHands receives 14 repository tasks; GPT Researcher
receives 14 public research questions. 
For each runtime, 4 tasks are used to train the forecasts and 10 to form the evaluation set. 
We use \texttt{gpt-5.6-luna} for model responses; GPT Researcher also uses \texttt{text-embedding-3-small} for embeddings.

Fig.~\ref{fig:wan-topology} summarizes the experimental replay pipeline. We first run each task with live model output and record when calls appear. After the run, we freeze the call sequence and exact request and response bytes. This creates a controlled workload on which we can run paired experiments: every compared controller receives the same workflow behavior. At the same time, each call remains hidden until replay reaches its release event. This replay pipeline allows us to isolate effects attributable to network policies rather than model variations in the AI-agent workflow responses.

Firstly, we answer the question: \textit{Do real AI-agent workflows expose the pre-connection control point needed in ASGE?} To answer the question, we performed a closer investigation on the runtimes of the 20 evaluation tasks, conducted from the pipeline in Fig.~\ref{fig:wan-topology}a. Results are presented in Fig.~\ref{fig:wan-occurrence}. All 20 evaluation tasks reveal at least one result-conditioned call that can be assigned to a replica and paths before connection (Wilson 95\% interval 0.839--1.000). At the call level, this holds for 75/87 OpenHands calls and 19/37 GPT Researcher calls, for a total of 94/124. The needed control point and the proposed opportunity, therefore, exist: The network controller can wait until it knows what service the workflow needs, then select the replica and paths before connection.

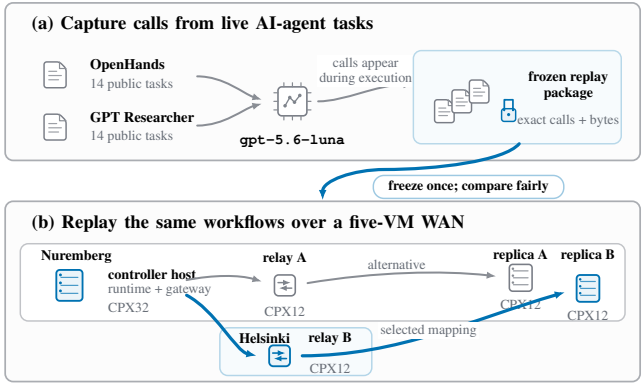
\begin{figure}[t]
\centering
\begin{tikzpicture}[
  font=\scriptsize,
  >=Latex,
  x=1mm,
  y=1mm,
  line cap=round,
  line join=round
]
\tikzset{
  panel/.style={draw=asgegray!52,rounded corners=3pt,line width=.65pt,
    fill=asgegray!2},
  region/.style={draw=asgegray!38,rounded corners=3pt,line width=.55pt,
    fill=white},
  helsinki/.style={draw=asgeblue!30,rounded corners=3pt,line width=.55pt,
    fill=asgeblue!3},
  flow/.style={-{Latex[length=1.3mm]},asgegray!72,line width=.76pt},
  selected/.style={-{Latex[length=1.4mm]},asgeblue,line width=1.2pt},
  smalllabel/.style={font=\fontsize{5.2}{5.8}\selectfont,
    text=asgegray!92,align=center},
  role/.style={font=\fontsize{5.3}{5.9}\selectfont\bfseries,
    text=black,align=center},
  pill/.style={draw=asgeblue!38,rounded corners=4pt,fill=white,
    line width=.55pt,font=\fontsize{5.2}{5.8}\selectfont\bfseries,
    text=black,inner xsep=2.0mm,inner ysep=.7mm}
}

\draw[panel] (0,29.2) rectangle (84,50);
\node[anchor=west,font=\scriptsize\bfseries] at (2.3,47.1)
  {(a) Capture calls from live AI-agent tasks};

\pic[scale=.62] at (6.6,40.5) {asge document};
\node[anchor=west,font=\fontsize{5.5}{6}\selectfont\bfseries]
  at (10.0,41.8) {OpenHands};
\node[anchor=west,smalllabel] at (10.0,39.2) {14 public tasks};

\pic[scale=.62] at (6.6,33.6) {asge document};
\node[anchor=west,font=\fontsize{5.5}{6}\selectfont\bfseries]
  at (10.0,34.9) {GPT Researcher};
\node[anchor=west,smalllabel] at (10.0,32.3) {14 public tasks};

\draw[flow] (25.4,40.3)
  .. controls (29.0,40.3) and (31.5,38.0) .. (34.0,37.2);
\draw[flow] (25.4,33.8)
  .. controls (29.0,33.8) and (31.5,36.0) .. (34.0,36.7);
\pic[scale=.78] at (38.0,37.0) {asge model};
\node[font=\fontsize{5.5}{6}\selectfont\bfseries\ttfamily,
  text=black,align=center] at (38.0,31.8)
  {gpt-5.6-luna};

\draw[flow] (41.4,37.0)
  .. controls (45.0,37.0) and (49.0,39.2) .. (54.0,39.2);
\node[smalllabel,fill=white,inner sep=.6pt] at (47.7,41.0)
  {calls appear\\during execution};

\draw[draw=asgeblue!28,fill=asgeblue!3,rounded corners=3pt,line width=.55pt]
  (54.0,31.4) rectangle (81.5,43.7);
\pic[scale=.55] at (58.1,35.7) {asge document};
\pic[scale=.55] at (60.4,37.0) {asge document};
\pic[scale=.55] at (62.7,38.3) {asge document};
\pic[scale=.60] at (66.6,35.4) {asge lock};
\node[role] at (74.5,39.1)
  {frozen replay\\package};
\node[smalllabel] at (74.5,34.4) {exact calls + bytes};

\draw[selected] (68.0,31.4)
  .. controls (66.0,29.1) and (57.0,28.0) .. (50.2,28.0)
  .. controls (46.0,28.0) and (42.0,26.5) .. (42.0,23.9);
\node[pill] at (61.3,25.7) {freeze once; compare fairly};

\draw[panel] (0,0) rectangle (84,23.8);
\node[anchor=west,font=\scriptsize\bfseries] at (2.3,20.9)
  {(b) Replay the same workflows over a five-VM WAN};

\draw[region] (2.0,8.0) rectangle (82.0,18.2);
\node[anchor=west,font=\fontsize{5.4}{6}\selectfont\bfseries]
  at (3.5,16.5) {Nuremberg};
\draw[helsinki] (28.4,.8) rectangle (48.5,7.1);
\node[anchor=west,font=\fontsize{5.4}{6}\selectfont\bfseries]
  at (29.7,5.7) {Helsinki};

\pic[scale=.68] at (8.5,12.7) {asge server selected};
\node[role,anchor=west] at (12.4,14.2) {controller host};
\node[smalllabel,anchor=west] at (12.4,12.1) {runtime + gateway};
\node[smalllabel,anchor=west] at (12.4,9.8) {CPX32};

\pic[scale=.52] at (37.0,12.7) {asge relay};
\node[role] at (37.0,16.3) {relay A};
\node[smalllabel] at (37.0,9.3) {CPX12};

\pic[scale=.59] at (68.2,13.7) {asge server};
\node[role] at (68.2,16.6) {replica A};
\node[smalllabel] at (68.2,10.2) {CPX12};
\pic[scale=.59] at (77.1,12.3) {asge server selected};
\node[role] at (77.1,16.6) {replica B};
\node[smalllabel] at (77.1,8.8) {CPX12};

\pic[scale=.52] at (36.3,3.4) {asge relay selected};
\node[role] at (43.0,5.7) {relay B};
\node[smalllabel] at (43.0,1.8) {CPX12};

\draw[flow] (24.0,13.4)
  .. controls (27.5,14.1) and (31.0,14.1) .. (34.0,13.2);
\draw[flow] (40.0,13.0)
  .. controls (49.0,14.5) and (57.8,14.7) .. (65.0,13.9);
\node[smalllabel,fill=white,inner sep=.7pt] at (51.8,15.5)
  {alternative};

\draw[selected] (24.0,11.4)
  .. controls (27.8,9.4) and (30.5,4.0) .. (33.3,3.4);
\draw[selected] (39.3,3.4)
  .. controls (52.0,3.4) and (62.0,8.3) .. (74.1,11.8);
\node[smalllabel,fill=white,inner sep=.7pt] at (56.0,6.7)
  {selected mapping};

\end{tikzpicture}
\caption{Main experimental pipeline. (a) We run public OpenHands and GPT
Researcher tasks with \texttt{gpt-5.6-luna}, then freeze the resulting calls
and exact bytes. (b) We replay each captured workflow under every controller
over five VMs connected by real TCP paths between Nuremberg and Helsinki. 
}
\label{fig:wan-topology}
\end{figure}

As an example, Fig.~\ref{fig:wan-control-point} shows one measured decision. We choose the evaluation block with the median call count and use its frozen identifier to break ties, ensuring the selection is independent of the outcome. In that call replay, one OpenHands response completes at $2.910\,\mathrm{s}$. The next call appears $60.3\,\mathrm{ms}$
later. In the WAN experiment setup, the controller binds it to the Helsinki relay and a Nuremberg replica in $8.1\,\mathrm{ms}$, and the gateway opens the connection $0.03\,\mathrm{ms}$ later.

\begin{figure}[t]
\centering
\includegraphics[width=.94\columnwidth]{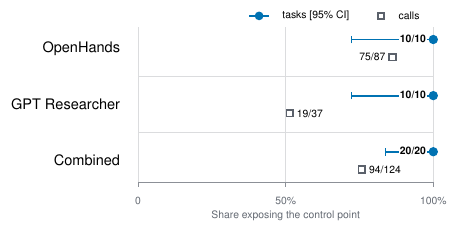}
\caption{Runtime-revealed control opportunity. Each row reports the share of
tasks with at least one result-conditioned call and the share of calls that are
bindable before connection. Filled circles show task shares with 95\% Wilson
intervals; open squares show call shares. Labels give exact counts.}
\label{fig:wan-occurrence}
\end{figure}

\begin{figure}[t]
\centering
\includegraphics[width=.98\columnwidth]{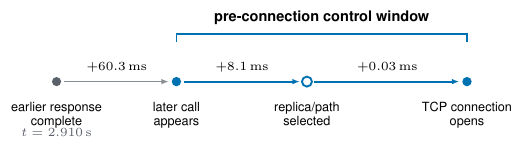}
\caption{Measured pre-connection control window. 
}
\label{fig:wan-control-point}
\end{figure}

\subsection{Validation under a controlled Docker environment}

Because reservations are central to ASGE-RR, we next ask whether keeping a preferred replica-and-path mapping available for a likely later call actually reduces that workflow's latency. We answer this question by comparing ASGE-RR with a matched no-reservation controller. 
Both controllers receive the same calls, forecasts, measurements, available replicas, and paths. The no-reservation controller chooses the lowest-delay mapping for the current call but does not protect capacity between calls. 

We run fixed versions of OpenHands SDK 1.28.1 and GPT Researcher 3.5.1. For each, one task makes a later call only after processing an earlier response. Each trial pairs this task with a direct task from the other framework, which creates competing traffic before the later call appears. 
The runtime processes execute and communicate through a nine-container Docker network, as shown in Fig.~\ref{fig:docker-validation}a. The containers implement a gateway, the controller, two forward relays, two  service replicas, two return relays, and an event recorder. Every dependency call uses a real HTTP connection through the gateway and the selected forward relay, service replica, and return relay.

We deliberately replay the same responses in every comparison so that both controllers receive identical application inputs and responses. The OpenHands responses are recorded during \texttt{gpt-5.6-luna} executions; the GPT Researcher responses are prevalidated deterministic fixtures. The actual runtimes still issue, receive, and process every HTTP transaction, allowing us to attribute the measured difference to the network decision rather than model sampling.

We test seven conditions: Moderate and high joint contention test the main mechanism at two load levels. Service-only and path-only contention separate the effects of replica queues and network paths. Abundant capacity and one legal mapping test cases in which there is no useful capacity choice to protect. Finally, a forecast miss measures the cost of reserving capacity for a call that never appears.

One trial executes the complete workflow scenario under one controller. Two runtime pairings, seven conditions, two controllers, and ten repetitions give 280 trials. We repeat the entire experiment with the controller order reversed, yielding 560 runs in total. Reversing the order verifies that warm-up or changes in host conditions do not explain the result.

Fig.~\ref{fig:docker-validation}b reports the results. Under moderate and high contention, ASGE-RR reduces the workflow's release-to-return latency by $1.45\,\mathrm{s}$ and $4.31\,\mathrm{s}$, respectively. The service-only and path-only conditions produce reductions of $1.53\,\mathrm{s}$ and $0.87\,\mathrm{s}$, showing that both replica queues and network paths can create an opportunity for reservations. The reversed-order experiment produces similar results. Under abundant capacity and with only one legal mapping, the intervals include zero: we detect no benefit when there is no useful alternative to protect. When the predicted later call does not appear, the unnecessary reservation instead delays the competing workflow by $0.46\,\mathrm{s}$.

\begin{figure}[t]
\centering
\includegraphics[width=.98\columnwidth]{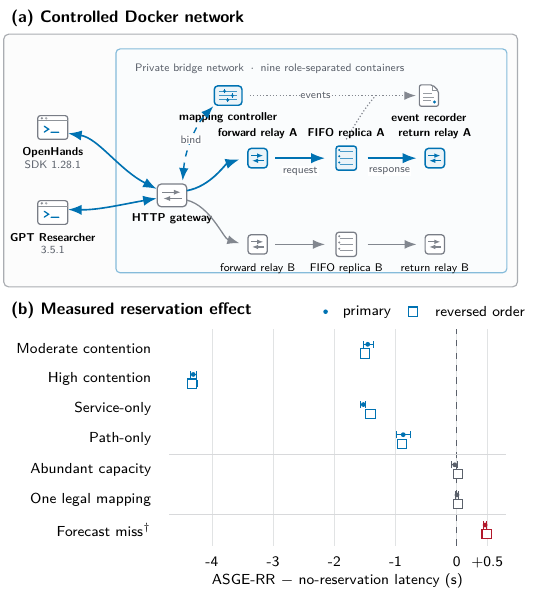}
\caption{
Controlled Docker environment. (a) The runtimes exchange HTTP transactions through nine role-separated containers. Before connection, the controller selects a forward relay, replica, and return relay. (b) ASGE-RR minus matched no-reservation release-to-return latency across seven conditions. Negative values favor ASGE-RR; the $^\dagger$forecast-miss row reports delay to the competing workflow. Each point summarizes 20 paired executions; bars show 95\% bootstrap intervals and open squares reverse controller order.
}
\label{fig:docker-validation}
\end{figure}

\subsection{Evaluation under a WAN testbed environment}

\subsubsection{The testbed}

Fig.~\ref{fig:wan-topology}b shows the setup of the testbed, which consists of a controller, two relay nodes, and two replicas, where the two relays are used to produce separate paths to the two replicas. 
Specifically, one CPX32 VM in Nuremberg runs the runtimes, which are replayed from the recorded runtime package, the gateway function, and the controller function. In addition, it is a recorder that records the performance. Two CPX12 relay VMs, one in Nuremberg and one in Helsinki, offer different paths. Two CPX12 service VMs in Nuremberg provide equivalent replicas. The controller chooses one of four relay-and-replica mappings before the gateway opens a TCP connection. Every request crosses VM interfaces and real TCP sockets, and Linux rate queues create contention from the paired workflow traffic. Two path contention levels are considered. A moderate-contention setting uses 4/8\, Mbit/s rates for the two paths, respectively; the saturation setting uses 2/4\, Mbit/s path rates.

\subsubsection{Confirmation of timing for control}

Recall the discussion about the pre-connection control point in Section~\ref{sec:evaluation.A}. To confirm this in the WAN setup, we measure the timing breakdown of ASGE-RR. The results for the moderate-contention case are shown in Fig.~\ref{fig:wan-timing} for example. We also include the timing breakdown for the rolling-horizon controller (described in the next subsection) as additional evidence. The result shows about $5.0\,\mathrm{s}$ of network transfer and queuing per moderate-contention replay, versus $1.18\,\mathrm{ms}$ of service time. ASGE-RR and rolling horizon decide in 5.8 and $6\,\mathrm{ms}$ on average. The controller therefore acts quickly relative to the network delay, leaving an opportunity to make network decisions for the workflows.

\begin{figure}[!t]
\centering
\includegraphics[width=.86\columnwidth]{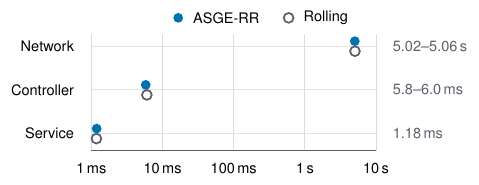}
\caption{Where experimental time is spent (moderate contention). 
}
\label{fig:wan-timing}
\end{figure}

\subsubsection{Performance comparison}

\begin{figure*}[!t]
\centering
\includegraphics[width=.98\textwidth]{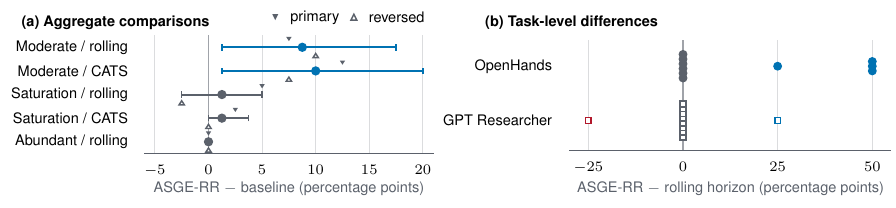}
\caption{Controller effects across operating conditions and evaluation tasks. (a) Circles and horizontal lines show the order-balanced differences and paired 95\% intervals; filled and open triangles show the two controller execution orders. (b) Each marker represents one evaluation block under moderate contention; stacked markers keep ties visible. Positive values favor \asgerr. Blue marks positive differences, red negative differences, and gray ties or intervals with zero.}
\label{fig:wan-results}
\end{figure*}

In Fig.~\ref{fig:wan-results}, we compare ASGE-RR with two strong, complementary baselines. The \textit{rolling-horizon baseline} follows the receding-horizon principle used in network control~\cite{avgeris2023mpc}. It gets the same forecasts and measurements as ASGE-RR, re-optimizes the priority-weighted completion objective when a call appears, and releases its forecast protection between events. The \textit{CATS-style baseline} represents established compute-aware traffic steering~\cite{li2026cats}. It treats each revealed call separately and selects the mapping with the earliest predicted finish. Rolling horizon tests whether retained reservations improve with repeated lookahead, while CATS tests whether workflow-aware control improves over strong current-call steering.

We report results for three experimental conditions: moderate contention, saturation, and abundant capacity. Each of the 20 evaluation tasks defines one evaluation block, and each comparison uses the same blocks and two controller execution orders. Reversing the order helps separate controller effects from warm-up or changes in the hosts and network. An independently implemented checker validates the inputs, mappings, payloads, and deadline outcomes of every reported replay.

Our primary metric is the fraction of offered workflows completed by the deadline. Each workflow has a value and contributes only if the entire workflow finishes on time. We pair controllers within each evaluation block and bootstrap the 20 blocks rather than treating individual calls as independent observations. Since the paired design holds the workflows, topology, forecasts, and configured capacity fixed across controllers, a potential gain of one controller over another, therefore, comes from the controller's decisions. The results are summarized in Fig.~\ref{fig:wan-results}. 

Fig.~\ref{fig:wan-results}a shows the aggregate comparisons. Specifically, under abundant capacity, ASGE-RR achieves no improvement. This is expected, because in this case, the two paths show no difference for ASGE-RR to exploit. When this difference appears, the improvement  emerges. Particularly, under saturation, compared with rolling horizon, ASGE-RR changes from 5.0 points better in the first order to 2.5 points lower when the order is reversed. The order-balanced difference is therefore 1.25 points, with an interval that includes zero (\(-2.5\)--5.0). The order-balanced difference against CATS is also 1.25 points and includes zero (0--3.75). 
Under moderate contention, ASGE-RR increases the fraction of workflow value completed by the deadline by 8.75 percentage points over rolling horizon (95\% paired interval: 1.25--17.5) and by 10 percentage points over CATS (1.25--20.0). The two controller execution orders give rolling-horizon differences of 7.5 and 10.0 points, respectively.

Fig.~\ref{fig:wan-results}b shows where that gain occurs. Each of the 20 evaluation tasks defines one evaluation block. ASGE-RR improves completed value in five blocks---four OpenHands blocks and one GPT Researcher block---ties rolling horizon in fourteen, and completes less value in one GPT Researcher block. These outcomes account for all 20 evaluation blocks.
The gain concentration matches the proposed mechanism. Reservations can change an outcome only when a current choice would otherwise consume capacity needed by a later call from another workflow to complete by  deadline. ASGE-RR preserves that option and revises it when the later call appears, without changing workflow reasoning or calls already issued.

\section{Related Work}
\label{sec:related}


\emph{Service graphs and embedding.}
SFC already supports service graphs, cycles, stateful forwarding, policy, and dynamic paths~\cite{halpern2015sfc}. ETSI VNF-FGs also represent graph and nested-graph services~\cite{etsi2014nfv}. VNF orchestration and VNE map functions and edges onto substrates with limited compute and bandwidth~\cite{bari2016orchestrating,chowdhury2009vne}, including affinity, trust, and path constraints~\cite{bouten2017affinity,torkzaban2020trust}. Online VNE admits requests one by one as complete objects~\cite{even2013onlinevne}. However, ASGE builds on the underlying infrastructure that has established network and resource placement and focuses on how to make use of them for workflows with runtime information. Most importantly, the service graph for a workflow in ASGE is unknown a priori, and consequently its embedding is dynamic and affected by any new call after the workflow has begun and earlier mappings have been committed.

\emph{Uncertainty-aware embedding and predictive control.}
Stochastic VNE represents uncertain link demand with scenarios and recourse~\cite{farkiani2019stochastic}. Model-predictive control (MPC) for SFC assurance uses a receding-horizon approach~\cite {avgeris2023mpc}. In contrast to ASGE-RR, later calls are predicted while their identities remain unrevealed, and earlier calls remain committed. 

\emph{Priority and deadline control.}
Prioritized dynamic-SFC systems admit, place, and reconfigure known chains under shared capacity~\cite{farkiani2021prioritized}. DCoflow admits and schedules coflows with deadlines~\cite{luu2022dcoflow}. ASGE-RR instead uses a dynamic reservation-priority score that changes with workflow value, predicted on-time probability, urgency, and capacity for future calls.

\emph{Steering and agent embedding.}
An individual ASGE mapping decision resembles Computing-Aware Traffic Steering (CATS), which uses compute and network metrics to select a service instance and overlay path~\cite{li2026cats}. The online ASGE setting is similar to that of AgentVNE, which maps agent nodes and virtual links under compute, memory, bandwidth, and affinity constraints and addresses topology changes~\cite{zheng2026agentvne}. 
Unlike them, the ASGE problem couples the mapping choices with fixed, earlier bindings, shared capacity, and deadline information. In addition, ASGE-RR selects replicas and forward and return paths for dependency calls under hard endpoint and transit policy as execution unfolds. 

\emph{AI-agent serving and orchestration.}
To serve AI-agent workflows, many serving engines and orchestrators have emerged. AIALTO routes variable-fan-out outputs through stateful compound-AI pipelines~\cite{santhanam2024alto}. Agentix schedules task graphs discovered at runtime~\cite{luo2026agentix}, while Nalar supports changing invocations, managed state, and adaptive instance routing~\cite{laju2026nalar}. Murakkab co-optimizes workflows, models, and hardware. Maestro uses cross-cluster RTT, model readiness, and memory pressure~\cite{chaudhry2026murakkab,wang2026maestro}. 
ASGE complements these systems by using runtime-revealed information to choose the replica and paths when a dependency call appears.

\section{Conclusion} 

AI-agent workflows increasingly interact with distributed models, memory systems, and tools whose invocation patterns are revealed only during execution. This runtime uncertainty creates a new networking challenge. We formalized this challenge as Agentic Service Graph Embedding (ASGE) and introduced a complete-information benchmark, assuming that future workflow evolution is known. To address the online setting, we presented ASGE-RR, an online ASGE controller that uses revisable reservations while preserving hard capacity and policy feasibility.
The experimental results with replayed real AI-agent workflows indicated that runtime-revealed AI-agent calls create a practical opportunity for networks to complete more workflows by deadline using existing capacity. This more broadly suggests that network control for agentic systems should account for evolving workflow dependencies rather than optimize each call in isolation.

\section*{Use of AI Disclosure}
Several AI tools were used in the experiments as introduced in the paper. In addition, OpenAI Codex was used to help implement ASGE-RR for the experiments. The authors remain responsible for all the results.

\balance

\end{document}